\documentclass[11pt,letterpaper]{article}
\RequirePackage{latexsym}
\RequirePackage{amssymb}
\RequirePackage{amstext}
\RequirePackage{amsmath}
\RequirePackage{amsthm}
\usepackage{ifthen}
\usepackage{fullpage}
\usepackage{times}
\usepackage{graphicx}
\usepackage{setspace}
\usepackage{hyperref}
\theoremstyle{plain}
\newtheorem{theorem}{Theorem}[section]
\newtheorem{lemma}[theorem]{Lemma}

\theoremstyle{definition}

\newtheorem{algorithm}[theorem]{Algorithm}

\newtheorem{program}[theorem]{Program}

\theoremstyle{remark}

\newenvironment{prg}[2]{
  \begin{samepage}
  \begin{program}
    \label{#1} $ $ \\[-11pt]
    \begin{itemize}
    \item [ ] #2 subject to:
      \begin{enumerate}
}{
      \end{enumerate}
    \end{itemize}
  \end{program}
  \end{samepage}
}

\def\bar{{\;|\;}}
\def\ep{{\varepsilon}}
\def\too{{\rightarrow}}

\newcommand{\maxx}[2]{{\mathrm{max}_{#1}\left\lbrace#2\right\rbrace}}

\newcommand{\en}[2]{{\mathrm{e}_{#1}\!\left(#2\right)}}

\begin{document}
\title{Scheduling under Precedence, Communication,\\and Energy
  Constraints\thanks{This material is based upon work supported by the National
    Science Foundation under Grant No. CCF-1016540.}}
\author{David Felber\footnote{University of California at Los Angeles,
    Department of Computer Science. Email: dvfelber@cs.ucla.edu.} \and Adam
  Meyerson\footnote{University of California at Los Angeles, Department of
    Computer Science. Email: awm@cs.ucla.edu.}}
\date{}
\maketitle
\begin{abstract}
We consider the problem of scheduling a set of $n$ tasks on $m$ processors under
precedence, communication, and global system energy constraints to minimize
makespan. We extend existing scheduling models to account for energy usage and
give convex programming algorithms that yield essentially the same results as
existing algorithms that do not consider energy, while adhering to a strict
energy bound.
\end{abstract}
\section{Introduction}
We consider the problem of scheduling a set of $n$ tasks on $m$ homogeneous
processors under precedence, communication, and global system energy constraints
to minimize makespan. This problem is of particular importance in portable and
embedded systems \cite{KAS2010,AAWMI2009,KR2009,DM1995}, where the
power demands of a growing number of computationally intensive
applications \cite{GVV2009,LTK2004} outmatch the growth rate of battery energy
density \cite{HKQPS1999}. It is also important in high-performance
systems and data centers \cite{SI2010} where the operational costs of
powering and cooling \cite{FKLB2008,HF2005,Gioiosa2010} and related
reliability issues from power dissipation \cite{WLDW2010,HF2005} are
substantial. Because multi-core processors are the
industry's answer to the power and thermal constraints limiting clock
speed \cite{RVC2007,Geer2005}, general scheduling methods that conserve energy
and minimize running times are needed to fully and efficiently use these systems.

The problem of scheduling tasks on processors to minimize the makespan (overall
runtime) has a long history \cite{Drozdowski2009}. In the most general form, we
are given a number of tasks to assign to processors, such that no processor may
run more than one task at a time and the goal is to minimize the makespan
(latest completion time). Distributing the tasks amongst processors typically
reduces the makespan; however the problem presents several issues which make the
solution non-trivial. First, a system has a finite number of processors which is
typically much less than the number of tasks, implying that some tasks must be
allocated to the same processor (thus potentially delaying their completion).
Second, there is communication between tasks in the form of precedence
constraints. When one task requires the output of another, we cannot schedule
them to run in parallel on different processors. Third, there is communication
between processors; the implication is that if the output of task $i$ is
required by task $j$ which is scheduled to run on a different processor, then
there will be some additional time delay $c_{ij}$ after the completion of $i$
for the necessary information to arrive. Note that this delay can be avoided by
scheduling $i,j$ on the same processor.

We further consider a overall bound on the total energy consumed by all tasks.
Our goal is to produce a smooth tradeoff between the total energy consumption
and the makespan, which is achievable by varying the energy bound. We observe
that some previous scheduling models permitted re-computation of tasks
(computing the same task multiple times on different processors to save on
communication delays); this is energy-wasteful and our models will assume that
such re-computation does not occur. Previous work on energy often tried to
reduce the energy consumption without any reduction to the makespan; this
approach has some merits but essentially treats energy as a ``secondary''
objective rather than producing a tradeoff between energy and makespan. Our
model also extends previous work by allowing the tradeoff between running time
and energy to be arbitrary (but convex) and task-dependent; this makes sense in
the context of tasks which heavily load different system components (i.e.
processor, network card, memory) and therefore may behave differently under
speed-scaling.

\paragraph{Problem Definition}
We formally state the problem as follows. We are given (1) a set $J$ of $n$
tasks; (2) $m$ processors; (3) a directed acyclic graph $G$ on $n$ nodes
representing precedence constraints between tasks; (4) a communication delay
$c_{ij}$ for each edge $(i,j)$ in $G$; (5) a set of energy/time tradeoff
functions $\en{j}{d_j}$; and (6) an energy bound $E$.

The goal is to construct a schedule $\sigma = \{\langle{}j,I,p\rangle{}\}$ where
task $j$ is assigned over time interval $I$ to processor $p$, such that (a) if
there is an edge from $i$ to $j$ in $G$, then task $i$ must finish before $j$
can begin; (b) each processor can only work on one task at any time; (c) $d_j$
is the total time for which task $j$ runs over all assigned intervals; (d) the
total energy used, $\sum_j \en{j}{d_j}$, is bounded by $E$; and (e) the time at
which the last task completes (the \textit{makespan}) is minimized.

When the communication delays are non-zero, we restrict attention to schedules
that do not migrate tasks, and require the additional constraint (f) that if
$(i,j)$ is an edge in $G$ and $i$ and $j$ are scheduled on different processors,
then $j$ cannot start until $c_{ij}$ time after $i$ has completed.

\paragraph{Previous Work}
Most previous work is experimental and considers the problem in the context of
dynamic voltage scaling, for which speed (operating frequency) is approximately
proportional to the voltage and for which power is approximately proportional to
voltage cubed \cite{SA2001}. This work includes more recent heuristics that
minimize makespan given a hard energy bound \cite{KAS2010,AAWMI2009} as well as
heuristics that minimize energy given hard timing constraints
\cite{SI2010,WTLC2010,GVV2009,KR2009,FKLB2008}. For both variations, the
approach generally taken is to create an initial schedule in which tasks are
scheduled at the highest speed possible, and then to reduce the speed of tasks
in such a way that the schedule length does not increase.

Several heuristic approaches also incorporate mathematical programming. Rountree
et al \cite{RLFFSS2007} use a linear program to bound the optimal solution from
below. Kang and Ranka \cite{KR2008a} and Leung et al \cite{LTK2004} use
heuristics to approximate the optimal solutions to integer programs. Zhang et al
\cite{ZHC2002} use a mathematical program to perform voltage scaling after the
tasks have been scheduled. Unlike our methods, these approaches do not yield
schedules with provable guarantees.

The only previous work that yields schedules with provable
guarantees while considering energy is by Pruhs et al \cite{PSU2008}. They develop a speed-scaling processor model for precedence-constrained tasks with three primary results. The
first is a proof that running all processors at a single fixed speed is at best
$\Omega(poly(m))$ approximate. The second is a proof that the total power across
all processors is constant over time in any optimal solution. The third result
is an algorithm that is $O(1)$ approximate in energy and $O(\log^{1+2/\alpha}
m)$ approximate in the makespan for a model-dependent constant $\alpha$. This
algorithm reduces the problem for $m \le n$ processors to the problem of
scheduling tasks on $m$ related processors (considered by Chekuri and Bender in
\cite{CB2001}) that run at speeds that are a function of a power level $p$,
which is chosen using a binary search to give the energy bound. We improve on
this result in section \ref{sec:nocomm-m}, where we combine convex programming
with the list scheduling result by Graham \cite{Graham1966,Graham1969} to obtain
an algorithm that is $(2-1/m)$-approximate in the makespan and satisfies the
energy bound exactly.

Algorithms for precedence-constrained task scheduling that give provable
guarantees on makespan (but without accounting for energy
considerations) have been an active area of research since the last millennium.
Graham \cite{Graham1966,Graham1969} gives a $(2-1/m)$-approximation algorithm
for the $m < n$ case and Fujii et al \cite{FKN1969} and Coffman and Graham
\cite{CG1972} give exact algorithms for the $m = 2$ case that use preemption and
migration. Papadimitriou and Yannakakis \cite{PY1988,PY1990} give a
$2$-approximate algorithm for the $m \ge n$, common communication delay $c_{ij}
= \tau$ case that uses recomputation, and Jung et al \cite{JKS1993} give an
exact algorithm for this case that is exponential in fixed integer $\tau$.
Munier and Konig \cite{MK1997} give an algorithm for $m \ge n$ and small
communication delays (where the duration of a task is at least $\rho \ge 1$
times any communication delay) that is $\beta$-approximate for $\beta =
\frac{2+2\rho}{1+2\rho}$, and Hanen and Munier \cite{HM1995,HM2001} give an
algorithm for the $m < n$ case that is $(1+\beta(1-1/m))$-approximate. A recent
survey by Drozdowski \cite{Drozdowski2009} details algorithms and heuristics for
the scheduling problem with communication delays.

\paragraph{Our Contributions}
Our primary contribution is to show that convex programming formulations of many
scheduling problems permit energy constraints to be simply appended; this enables us to produce comparable provable results to the energy-blind case without substantially more complex analysis. In section \ref{sec:nocomm} we consider the case where
the communication delays are all zero. We obtain optimal schedules when $m \ge
n$ and for $m = 2$; and for $m < n$, we obtain $(2-1/m)$-approximate schedules.
These results are analogous to the energy-blind results of
\cite{Graham1966,Graham1969} and \cite{FKN1969,CG1972}, and improve on the
result in \cite{PSU2008}.

In section \ref{sec:smcomm} we consider the case where the communication delays
may be non-zero. For small communication delays we obtain $\beta$-approximate
schedules for $m \ge n$ processors, $\beta = \frac{2+2\rho}{1+2\rho}$, and
$(1+\beta(1-1/m))$-approximate schedules for arbitrary $m$. These results extend
the energy-blind results in \cite{HM1995,MK1997,HM2001}. For large communication
delays we extend the approach of \cite{GKMP2008} to obtain
$\frac{2(R+1)}{3}$-approximate schedules. In these cases, $\rho$ and $R$ are
parameters to the algorithm that bound the relative size of the delays.

\paragraph{Discussion}
The variants for which we obtain approximations are NP-complete
\cite{Ullman1975,RaywardSmith1987,PY1988} and are therefore unlikely to have
fast exact solutions. We obtain our results modulo $\ep$ error due to the finite
precision used in solving convex programs. To simplify our analyses we do not
mention this term further.

Except in section \ref{sec:nocomm-twoproc}, our algorithms do not use preemption
(stopping and starting of tasks), migration (moving tasks from one processor to
another), or recomputation (computing a task more than once.) Our algorithms in
section \ref{sec:nocomm} are approximate even to optimal algorithms that do have
these properties. Our algorithms in section \ref{sec:smcomm} are approximate to
optimal algorithms that do not have these properties.

We consider cases in which $e_j$ and $d_j$ are inversely related according to a
convex function; that is, $\en{j}{d_j}$ is convex and non-increasing. It is
natural to make this assumption for several reasons. Convexity is a more general
assumption than either the speed scaling model in \cite{PSU2008} or the dynamic
voltage scaling model assumed by the experimental work cited above; the
time/energy tradeoff in both of these models is convex. We consider separate
functions for each task because they may vary in their use of resources and
therefore may not have the same curve, even when a measure of ``work'' for a
task is taken into account. Lastly, because we consider a homogeneous system in
which there is no contention for resources (other than processors) by tasks, the
functions are independent of processors and of other tasks.

\paragraph{Notation}
We use the following notation throughout the paper. $i$, $j$, and $k$ are tasks.
$p$ and $q$ are processors. $p_j$ is the processor on which $j$ is scheduled. We
write $i \too j$ if the edge $(i,j)$ exists in $G$, $i < j$ if there is a
(directed) path from $i$ to $j$ in $G$, and $i \sim j$ if neither $i < j$ nor $j
< i$. The duration of a task $j$ is $d_j$ and the energy it consumes is $e_j$;
these are inversely related in a model-dependent manner. $t_j$ is the start time
of task $j$. $\sigma$ is a schedule, $\mu$ is a makespan and $E$ is the energy
bound. We write $t_j(\sigma_1)$, $\mu(\sigma_2)$, etc. to differentiate between
schedule values. Task $j$ is a \textit{source} task if there is no $i \too j$
and a \textit{sink} task if there is no $j \too k$. We say a task $j$ is
\textit{active} if it is currently running on a processor and \textit{available}
if its predecessors have all completed (regardless whether it's active.) We say
a processor is \textit{active} at time $t$ if it is working on a task at time
$t$, else \textit{idle}.

All of these parameters are required to be non-negative. We abuse notation and
use $I$ to refer to both an interval $I = [a,b]$ and its length $b - a$, and we
use $S$ to refer to both the set $S$ and the number of elements it contains. If
$S$ is a set of disjoint intervals, we also write $S$ for $\sum_{I \in S} I$.

\paragraph{Organization}
We organize the rest of the paper as follows. Section \ref{sec:nocomm} contains
our results for the case in which the communication delays are all zero. In
section \ref{sec:smcomm} we extend the results of \cite{HM1995,MK1997,HM2001} to
the case in which there are small communication delays and the result of
\cite{GKMP2008} to the case in which the delays are large.
\section{Zero Communication Delays}
\label{sec:nocomm}
In this section we consider a model in which a task $j$ may be started as soon
as all $i \too j$ have completed and there is a free processor. Our algorithms
in this section are competitive against preemptive and migratory algorithms as a
result of lemma \ref{thm:nocomm-mgen}. Recomputation never helps in this model;
if a schedule $\sigma$ that computes a task $j$ more than once removes the
second computation, the total energy will decrease, and every task can still
begin at $t_j(\sigma)$. We therefore assume that no schedule uses recomputation.
\subsection{$m \ge n$ Processors}
\label{sec:nocomm-mgen}
We use the following convex program in the $m \ge n$ case. We assume that
$\en{j}{d_j}$ can be computed in polynomial time.
\begin{prg}{prg:nocomm-mgen}{Minimize $\mu$}
\item $t_j \ge t_i + d_i$ for all $i \too j$
\item $\mu \ge t_j + d_j$ for all $j$
\item $\sum_j \en{j}{d_j} \le E$
\item \label{itm:avg-nocomm-mgen} $\mu \ge \frac{1}{m} \sum_j d_j$
\item [ ] $\mu, t_j, d_j \ge 0$
\end{prg}
Constraint \ref{itm:avg-nocomm-mgen} will be necessary in section
\ref{sec:nocomm-m}. In the case where $m \ge n$, this only constrains the
makespan to be at least the average duration.

We show that this convex program is equivalent to the scheduling problem.
\begin{theorem}
\label{thm:nocomm-mgen}
There is a solution $\mu, t_j, d_j$ to program \ref{prg:nocomm-mgen} iff there
is a feasible schedule $\sigma$ with makespan $\mu$.
\end{theorem}
\begin{proof}
We prove the theorem by constructing a solution to one out of a solution to the
other. The following two algorithms perform these conversions.
\end{proof}
\begin{algorithm}
\label{alg:nocomm-mgen-prg2sig}
Given a solution $\mu, t_j, d_j$ to program \ref{prg:nocomm-mgen}, schedule each
task $j$ in $\sigma$ at time $t_j$ with duration $d_j$. Always schedule $j$ on
its own processor.

Precedence and energy constraints are satisfied immediately. $j$ does not cause
co-occurrence conflicts with any other task $k$ because it is on its own
processor. $\sigma$ uses $n \le m$ processors.
\end{algorithm}
\begin{algorithm}
\label{alg:nocomm-mgen-sig2prg}
Given a schedule $\sigma$ on $m$ processors, construct $\mu, t_j, d_j$ as
follows. Set $\mu = \mathrm{makespan}(\sigma)$ and $t_j = t_j(\sigma)$. Let
$S_j$ be the set of intervals over which $j$ is active in $\sigma$. Set $d_j =
S_j$.

Constraint $1$ is satisfied because $t_j(\sigma)$ is at least the end of the
last interval in $S_i$ for $i \too j$, and the first interval in $S_i$ does not
begin until $t_i(\sigma)$. Constraint $2$ is satisfied similarly. Constraint $3$
is satisfied because $\sigma$ is feasible, and constraint $4$ is satisfied
because $\sigma$ uses at most $m$ processors and the makespan is at least the
average load.
\end{algorithm}
Because program \ref{prg:nocomm-mgen} is a convex optimization problem with
polynomially many constraints, it can be solved in polynomial time. To generate
an optimal schedule $\sigma$ we first solve program \ref{prg:nocomm-mgen} and
then use algorithm \ref{alg:nocomm-mgen-prg2sig} to construct the schedule.
\subsection{$m$ Processors}
\label{sec:nocomm-m}
In the case where $m$ is arbitrary we obtain a $(2-1/m)$-approximation.
\begin{lemma}
\label{lem:nocomm-m}
Given a schedule $\sigma_1$ that satisfies program \ref{prg:nocomm-mgen},
algorithm \ref{alg:nocomm-m} constructs a schedule $\sigma_2$ that has makespan
$\mu_2 \le (2-1/m) \, \mu_1$ and that uses at most $E$ energy.
\end{lemma}
\begin{algorithm}
\label{alg:nocomm-m}
While there are unscheduled tasks in $\sigma_2$, let $t$ be the earliest time
for which there is both a processor $p$ that has no tasks scheduled after time
$t$ and a task $j$ whose predecessors $i \too j$ have all completed by time $t$.
Schedule $j$ for duration $d_j(\sigma_1)$ starting at time $t$ on processor $p$
in $\sigma_2$.
\end{algorithm}
\begin{proof}
We cut the time interval $(0,\mu_2)$ at each point at which some task begins or
ends, and partition these sub-intervals into two sets $A$ and $B$, where $A$
contains all intervals in which all $m$ processors are active in $\sigma_2$ and
$B$ contains the rest. $\mu_2 = A + B$. We define $W = \sum_j d_j$ for
convenience.

We bound $B \le \mu_1$ with a potential argument. For each time $t$ in
$\sigma_2$ let $F_t$ be the set of tasks finished by time $t$ and $J_t$ the set
of available tasks at time $t$. We define $\phi(t)$ as the smallest time $u$ in
$\sigma_1$ such that $\sigma_1$ finishes all tasks $F_t$ by time $u$ and has
completed at least as much work on all tasks in $J_t$ as $\sigma_2$ has
completed. We cannot have $\phi(t) > \mu_1$ since $\sigma_1$ completes all tasks
by time $\mu_1$. We must also have $\phi(0) = 0$ and $\phi(t_1) \le \phi(t_2)$
whenever $t_1 \le t_2$.

For each interval $I = (a,b)$ in $B$ it must be that all $J_a = J_b < m$
available tasks are active; otherwise, we could have scheduled an available task
on an idle processor. We must further have $\phi(b)-\phi(a) \ge b-a$; otherwise,
$\sigma_1$ could complete $I \, J_a$ work on these $J_a$ tasks in less than $I$
time. Together, these yield
\begin{align*}
B &= \sum_{(a,b) \in B} b - a \le \sum_{(a,b) \in B} \phi(b)-\phi(a) \\ &\le
\sum_{(a,b) \in A \cup B} \phi(b)-\phi(a) = \phi(\mu_2) - \phi(0) \le \mu_1
\end{align*}
$\sigma_2$ completes at least $I$ work over each interval $I$ in $B$ because in
each such $I$ there is at least one active processor; otherwise, $t$ would be
the earliest time at which we could schedule a task in $\sigma_2$, but we chose
a time $t' > t$ instead.

We bound $A \le (W-B)/m$ since there is no more than $(W-B)$ work to do in $A$
and all $m$ processors are active in $A$.

By constraint \ref{itm:avg-nocomm-mgen} in program \ref{prg:nocomm-mgen}, $\mu_1
\ge W/m$. Together with our bounds for $A$ and $B$ we have $\mu_2 = A + B \le
W/m + (1-1/m) \, B \le (2-1/m) \, \mu_1$.
\end{proof}
\begin{theorem}
We can construct a $(2-1/m)$-approximation for the case when $m$ is arbitrary.
\end{theorem}
\begin{proof}
Let $\mu^*$ be the optimal makespan obtainable for $m$ processors on a given
task graph $G$ with energy bound $E$, using an optimal schedule $\sigma^*$.
Program \ref{prg:nocomm-mgen} has an objective value $\mu_1$ that is no larger
than $\mu^*$; otherwise, we could set $t_j = t_j(\sigma^*)$ and $d_j =
d_j(\sigma^*)$, and all constraints would remain satisfied.

Since $\mu_1 \le \mu^*$, we can can use algorithm \ref{alg:nocomm-mgen-prg2sig}
to create a schedule $\sigma_1$ that uses $n$ processors, and algorithm
\ref{alg:nocomm-m} to generate a schedule $\sigma_2$ that uses $m$ processors.
By the lemmas, we have $\mu(\sigma_2) \le (2-1/m) \, \mu_1 \le (2-1/m) \,
\mu^*$.
\end{proof}
\subsection{Two Processors}
\label{sec:nocomm-twoproc}
In the case where $m = 2$ we can obtain an optimal schedule if migration and
preemption are permitted. We combine convex programming with a fractional
version of the matching-based algorithm in \cite{FKN1969}.

We consider tasks $j$ according to any linear ordering $<_L$ that is an
extension of $<_G$. The duration $d_j$ of a task $j$ is broken into components.
The time when task $j$ is the only task running is $\ell_i$. The time when task
$j$ is running at the same time as another task $i <_L j$ is $\ell_{ij}$.
\begin{prg}{prg:nocomm-twop-lp}{Minimize $\mu = \sum_j \ell_j + \sum_i \sum_{j >_L i} \ell_{ij}$}
\item \label{itm:dj-nocomm-twop} $d_j = \ell_j + \sum_{i <_L j} \ell_{ij} +
  \sum_{k >_L j} \ell_{jk}$ for all $j$
\item \label{itm:prec-nocomm-twop} $\ell_{ij} = 0$ for all $i <_G j$
\item [ ] $\ell_j, \ell_{ij} \ge 0$
\end{prg}
We extend this linear program with the addition of a convex energy constraint.
We consider $d_j$ to be a fixed value in program \ref{prg:nocomm-twop-lp} and a
variable in program \ref{prg:nocomm-twop}.
\begin{program}
  \label{prg:nocomm-twop} $ $ \\[-11pt]
  \begin{itemize}
  \item [ ] Program \ref{prg:nocomm-twop-lp} subject to the additional constraint:
    \begin{enumerate}
      \setcounter{enumi}{2}
    \item \label{itm:en-nocomm-twop} $\sum_j \en{j}{d_j} \le E$
    \end{enumerate}
  \end{itemize}
\end{program}
We can construct a solution to program \ref{prg:nocomm-twop} from a feasible
schedule $\sigma$ by setting $\ell_j$ ($\ell_{ij}$) as the sum over intervals in
which only task $j$ (tasks $i$ and $j$) are active.

To construct a feasible schedule from a solution to the program, we use the
following lemma as a subroutine.
\begin{lemma}
\label{lem:nocomm-twop}
Let $S$ be the set of source tasks in $G$. Given a solution to program
\ref{prg:nocomm-twop-lp} with objective value $\mu$, we can construct a new
solution with the same objective value that satisfies the condition $\exists s
\in S \; \forall j \notin S \left(\ell_{sj} = 0\right)$; call this $s$ the
\textit{next} task.
\end{lemma}
\begin{proof}
Let $Z = \left\{(i,j) \bar i \in S, j \notin S, \ell_{ij} > 0\right\}$. We say
$Z$ has inductive pairs $(i,j)$ and $(h,k)$ if these pairs exist in $Z$ and $j
\sim k$. The proof is by complete induction on the set $Z$.

If $Z$ is empty, then the lemma holds for all $j$ in $S$. If $Z$ does not have
inductive pairs, then we define $T = \left\{j \bar \exists i \left((i,j) \in
Z\right)\right\}$. It must be that the tasks in $T$ are fully ordered by $<_G$
(and not just $<_L$); otherwise, we would have $j \sim k$ for some pairs. Let
$j$ be the least task in $T$. Since $j$ is in $T$, it is not in $S$, and
therefore there is some $s$ in $S$ such that $s <_G j$. There can be no $k \in
T$ for which $(s,k) \in Z$; otherwise, $s <_G j <_G k$ would violate constraint
\ref{itm:prec-nocomm-twop}. Because $s$ does not appear in $Z$, the lemma holds
for $s$.

For the inductive case, let $(i,j)$ and $(h,k)$ be inductive pairs in $Z$, so
that $j \sim k$. We also have that $i \sim h$ because $i$ and $h$ are both in
$S$. We define $\Delta = \min \left\{\ell_{ij},\ell_{hk}\right\}$, and update
$\ell_{ih} := \ell_{ih} + \Delta$, $\ell_{jk} := \ell_{jk} + \Delta$, $\ell_{ij}
:= \ell_{ij} - \Delta$, and $\ell_{hk} := \ell_{hk} - \Delta$.

Neither $d_j$ nor $\mu$ changes as a result of these updates. No $\ell$ is
updated to be less than zero. Constraint 2 still holds because $j \sim k$ and $i
\sim h$. Therefore, this new solution is also feasible and has the same
objective value. Because $\Delta = \min \left\{\ell_{ij},\ell_{hk}\right\}$,
either $\ell_{ij} = 0$ or $\ell_{hk} = 0$, so $Z$ decreases and we continue
inductively.
\end{proof}
The following algorithm constructs a feasible schedule $\sigma$ iteratively.
Program \ref{prg:nocomm-twop} is solved once initially to determine $d_j$ that
satisfy constraint \ref{itm:en-nocomm-twop}. In each iteration, a solution is
maintained for program \ref{prg:nocomm-twop-lp} using fixed values chosen so
that the total durations are $d_j$.
\begin{algorithm}
\label{alg:nocomm-twop}
Initially, solve program \ref{prg:nocomm-twop} for the input graph $G$ to
determine durations $d_j$ for each task $j$. Define $G_0 = G$, $\hat{d}^0_j =
d_j$, and $\hat{\sigma}_0$ to be an empty schedule.

Then, for each $l = 0, \ldots, n-1$, do the following. Let $S_l$ be the set of
sources in $G_l$. Solve program \ref{prg:nocomm-twop-lp} on $G_l$ using fixed
durations $\hat{d}^l_j$. Run lemma \ref{lem:nocomm-twop} as a subroutine to find
the next task $s_l \in S_l$ and updated values $\ell^l_j, \ell^l_{ij}$. To get
$\hat{\sigma}_{l+1}$, begin with $\hat{\sigma}_l$, schedule task $s$ alone for
$\ell^l_s$ time, and then for each other $i \in G_l$ schedule tasks $s$ and $i$
together for $\ell^l_{si}$ time. Define $G_{l+1}$ as $G_l - \{s\}$ and
$\hat{d}^{l+1}_i = \hat{d}^l_i - \ell^l_{si}$.

Our resulting schedule $\sigma$ is $\sigma_n$.
\end{algorithm}
We note that the updated values $\ell^l_j, \ell^l_{ij}$ for all tasks except $s$
constitute a feasible optimum for program \ref{prg:nocomm-twop-lp} on $G_{l+1},
\hat{d}^{l+1}_j$, so by saving these values we do not need to recompute a
solution to this program.
\begin{theorem}
\label{thm:nocomm-twop}
Let $\Pi(G)$ be program \ref{prg:nocomm-twop-lp} instantiated for graph $G$.
$\sigma$ is feasible and has makespan $\mu(\sigma)$ equal to the objective value
$\mu^\Pi_0$ of $\Pi(G_0)$.
\end{theorem}
\begin{proof}
Let $\mu_l = \mu(\hat{\sigma}_l)$. We first prove that $\mu_{l+1} +
\mu^\Pi_{l+1} \le \mu_l + \mu^\Pi_l$ for all $l$. We have $\mu_{l+1} = \mu_l +
\ell^l_s + \sum_{i >_L s} \ell^l_{si}$ because we schedule $s$ and $(s,i)$ to
get $\hat{\sigma}_{l+1}$. We also have that $\mu^\Pi_{l+1} \ge \mu^\Pi_l -
\ell^l_{s} - \sum_{i >_L s} \ell^l_{si}$; otherwise, we could use $\ell^l_{j \ne
  s}, \ell^l_{ij \ne s} = \ell^{l+1}_{j}, \ell^{l+1}_{ij}$ and get a better
solution to $\Pi(G_l)$.

Inductively, we must have $\mu_{l+1} + \mu^\Pi_{l+1} \le \mu^\Pi_0$ because
$\mu_0 = 0$. We further have that $\mu^\Pi_n = 0$ because $G_n$ is empty.
Therefore, $\mu(\sigma) \le \mu^\Pi_0$. Because we can construct a feasible
solution to $\Pi(G_0)$ from $\sigma$, we also have $\mu(\sigma) \ge \mu^\Pi_0$.

We also must prove that $\sigma$ uses at most $E$ energy. The $d_j$ were chosen
by program \ref{prg:nocomm-twop} to satisfy constraint \ref{itm:en-nocomm-twop}.
Program \ref{prg:nocomm-twop-lp} satisfies constraint \ref{itm:dj-nocomm-twop}
with equality. Because we update $\hat{d}^{l+1}_i = \hat{d}^l_i - \ell^l_{si}$,
the total duration for which each task $j$ runs is $d_j$, and therefore the
total energy used is at most $E$.
\end{proof}
\section{Small Communication Delays}
\label{sec:smcomm}
In this section we consider a model in which each edge $i \too j$ in $G$ has an
associated communication delay of $c_{ij}$. When task $i$ finishes, if task $j$
is scheduled on a different processor, it may not begin until at least $c_{ij}$
time has passed. (If $j$ is started on the same processor, it may begin
immediately after $i$ finishes, as long as it is not otherwise constrained.)
This models a system in which there is a time/energy trade-off for processors
and a constant-speed interconnect running at $s$ data per time that must
transfer $s \, c_{ij}$ data from $p_i$ to $p_j$ and can make any number of
point-to-point transfers at a time.

Because we consider small communication delays, we assume a given $\rho \ge 1$
for which $c_{ij} \le d_k / \rho$ for all $i \too j$ and $k$, and we compare
against optimal solutions that also satisfy this constraint. Our results improve
with increasing $\rho$ (and smaller communication delays.)
\subsection{$m \ge n$ Processors}
We extend the approach of \cite{HM1995,MK1997} to account for energy. We first
show that the following non-convex program is equivalent to the scheduling
problem. $x_{ij}$ is an indicator variable that is $1$ if $j$ follows $i$ on
processor $p_i$ without waiting the $c_{ij}$ communication delay time, and $0$
otherwise.
\begin{prg}{prg:smcomm-mgen-nc}{Minimize $\mu$}
\item $t_j \ge t_i + d_i + (1-x_{ij}) \, c_{ij}$ for all $i \too j$
\item $\sum_{j : i \too j} x_{ij} \le 1$ for all $i$
\item $\sum_{i : i \too j} x_{ij} \le 1$ for all $j$
\item $\mu \ge t_j + d_j$ for all $j$
\item $x_{ij} \in \{0, 1\}$ for all $i \too j$
\item $\sum_j \en{j}{d_j} \le E$
\item $c_{ij} \le d_k / \rho$ for all $i \too j$, $k$
\item \label{itm:avg-smcomm-mgen} $\mu \ge \frac{1}{m} \sum_j d_j$
\item [ ] $\mu, t_j, d_j \ge 0$
\end{prg}
\begin{lemma}
\label{lem:smcomm-mgen-nc}
There is a solution $\mu, t_j, d_j, x_{ij}$ to program \ref{prg:smcomm-mgen-nc}
iff there is a feasible schedule $\sigma$ with makespan $\mu$ that satisfies
$d_i / \rho \ge c_{ij}$ and that does not use preemption, migration, or
recomputation.
\end{lemma}
\begin{proof}
We prove the lemma by constructing a solution to one out of a solution to the
other. The following two algorithms perform these conversions.
\end{proof}
\begin{algorithm}
\label{alg:smcomm-mgen-prg2sig}
Given a solution $\mu, t_j, d_j, x_{ij}$ to program \ref{prg:smcomm-mgen-nc},
construct $\sigma$ as follows. Consider tasks $j$ according to any linear
extension $<_L$ of $<_G$. Schedule each task $j$ at time $t_j$ with duration
$d_j$. If there is some $i \too j$ for which $x_{ij} = 1$ schedule $j$ on $p_i$
and on a new processor otherwise.

Precedence, communication, and energy constraints are satisfied immediately. If
$j$ is scheduled on a new processor, then $j$ does not cause co-occurrence
conflicts with any prior task $k <_L j$ because it is on a new processor. If $j$
is scheduled on $p_i$, then it could co-occur with a prior task $k <_L j$ only
if $k$ is also scheduled on $p_i$ after $i$. But $k$ is not scheduled on $p_i$
after $i$ because $x_{ij} = 1$ and therefore $x_{ik} = x_{kj} = 0$ by
constraints $2$ and $3$, so in this case $j$ also does not cause co-occurrence
conflicts.
\end{algorithm}
\begin{algorithm}
\label{alg:smcomm-mgen-sig2prg}
Given a schedule $\sigma$, construct $\mu, t_j, d_j, x_{ij}$ as follows. Set
$\mu = \mathrm{makespan}(\sigma)$, $t_j = t_j(\sigma)$, $d_j = d_j(\sigma)$. Set
$x_{ij} = 1$ if $t_j < t_i + d_i + c_{ij}$ and $x_{ij} = 0$ otherwise.

All constraints except for $2$, $3$, and $8$ are satisfied immediately. Suppose
$x_{ij} = 1$. Since $t_j < t_i + d_i + c_{ij}$ and $c_{ij} \le d_k$, it must be
that no other task $k$ is scheduled between $t_i + d_i$ and $t_j$ on processor
$p_i$. Therefore, for any other task $k$ where $i \too k$ or $k \too j$, $x_{ik}
= x_{kj} = 0$, so constraints $2$ and $3$ hold. Lastly,
$\mathrm{makespan}(\sigma)$ is at least the average load on each processor, and
therefore constraint \ref{itm:avg-smcomm-mgen} holds.

We note that we need $\rho \ge 1$ to enforce that no task $k$ can be between $i$
and $j$ on $p_i$; if $\rho < 1$ then we could possibly have $d_k < c_{ij}$,
which would break our argument.
\end{algorithm}
We relax program \ref{prg:smcomm-mgen-nc} to a convex program by requiring only
that $0 \le x_{ij} \le 1$ rather than $x_{ij} \in \{0, 1\}$. We show that the
following deterministic rounding algorithm gives a feasible schedule $\sigma$
that satisfies the energy bound $E$ and is $\beta$-approximate in the makespan
for $\beta = \frac{2+2\rho}{1+2\rho}$.
\begin{algorithm}
\label{alg:smcomm-mgen}
Solve the convex relaxation of program \ref{prg:smcomm-mgen-nc}. For each
$x_{ij}$, Define $\hat{x}_{ij} = 1$ if $x_{ij} > 1/2$ and $\hat{x}_{ij} = 0$
otherwise. Take any linear extension $<_L$ of $<_G$. For each $j$ in order of
$<_L$, if there is an $i \too j$ such that $\hat{x}_{ij} = 1$, schedule $j$ on
$p_i$ to begin at time $t_i(\sigma) + d_i$. If there is no such $i$, schedule
$j$ on its own processor at time $\maxx{i \too j}{t_i(\sigma) + d_i + c_{ij}}$.
Always schedule $j$ for duration $d_j$.
\end{algorithm}
\begin{proof}
We first prove the stronger claim that for every task $j$, $t_j(\sigma) \le
\beta \, t_j$. We prove this inductively on $<_G$.

If $j$ is a source task, then there are no $i \too j$, so $t_j(\sigma) = 0 \le
\beta \, t_j$.

If $j$ is not a source task, then $j$ is scheduled either on $p_i$ for some $i
\too j$ or on its own processor. If $j$ is scheduled on $p_i$ then $\hat{x}_{ij}
= 1$, and therefore $t_j(\sigma) = t_i(\sigma) + d_i \le \beta \, t_i + d_i \le
\beta \, t_j$ by induction and constraint 1. Otherwise, there is some $i \too j$
for which $j$ is scheduled at time $t_j(\sigma) = t_i(\sigma) + d_i + c_{ij} \le
\beta \, t_i + d_i + c_{ij}$. $\hat{x}_{ij} = 0$, so $x_{ij} \le 1/2$, and
therefore $1 - x_{ij} \ge 1/2$ and $t_j \ge t_i + d_i + c_{ij}/2$. With $\beta =
\frac{2+2\rho}{1+2\rho}$ and $c_{ij} \le d_i/\rho$, algebraic manipulation
yields $t_j(\sigma) \le \beta t_j$.

With this stronger claim, we can bound the makespan by
\[
\mu(\sigma) = \maxx{j}{t_j(\sigma) + d_j} \le \maxx{j}{\beta \, t_j + d_j} \le
\beta \, \mu
\]
We now prove that $\sigma$ is a feasible schedule. Precedence, communication,
and energy constraints are satisfied immediately. If $j$ is scheduled on a new
processor, then $j$ does not cause co-occurrence conflicts with any prior task
$k <_L j$ because it is on a new processor. If $j$ is scheduled on $p_i$, then
$\hat{x}_{ik} = \hat{x}_{kj} = 0$ for any other task $k$, so in this case $j$
also does not cause co-occurrence conflicts.
\end{proof}
\subsection{$m$ Processors}
We use the following lemma from \cite{HM2001} as a black-box to obtain a bound
for $m$ processors.
\begin{lemma}
\label{lem:smcomm-m-hm}
In the case where $c_{ij} \le d_k/\rho$, given a schedule $\sigma_1$ with
makespan $\mu_1$ that uses more than $m$ processors we can construct a new
schedule $\sigma_2$ with makespan $\mu_2 \le \frac{1}{m} \sum_j(d_j) + (1-1/m)
\, \mu_1$ that uses only $m$ processors.
\end{lemma}
The model considered in \cite{HM2001} assumes fixed durations, so $d_j(\sigma_1)
= d_j(\sigma_2)$ and therefore $\sigma_2$ uses at most $E$ energy as well.
\begin{theorem}
\label{thm:smcomm-m}
We can construct a $(1+\beta(1-1/m))$-approximate schedule $\sigma$ for the case
in which $m < n$.
\end{theorem}
\begin{proof}
We use algorithm \ref{alg:smcomm-mgen} to generate a schedule $\sigma_1$ with
makespan $\mu_1 \le \beta \, \mu$, where $\mu$ is the optimal objective value of
program \ref{prg:smcomm-mgen-nc}.

Let $\sigma^*$ be an optimal schedule for $m$ processors with makespan $\mu^*$.
We must have $\mu^* \ge \mu$; otherwise, we could use algorithm
\ref{alg:smcomm-mgen-sig2prg} to convert $\sigma^*$ into a feasible solution to
program \ref{prg:smcomm-mgen-nc} with smaller $\mu$.

$\sigma^*$ satisfies constraint \ref{itm:avg-smcomm-mgen} of program
\ref{prg:smcomm-mgen-nc} and therefore $\mu^* \ge \frac{1}{m} \sum_j d_j$. These
three bounds plus the bound in lemma \ref{lem:smcomm-m-hm} yield the theorem.
\end{proof}
\subsection{Large Communication Delays}
We use the approach of \cite{GKMP2008} to obtain a
$\frac{2(R+1)}{3}$-approximate schedule $\sigma_2$ in the case of large
communication delays where $c_{ij} / d_k \le R$.
\begin{algorithm}
\label{alg:lgcomm-mgen}
Define $\hat{c}_{ij} = c_{ij} / R$. Run algorithm \ref{alg:smcomm-mgen} on $G$,
$\hat{c}_{ij}$, and $\rho = 1$ to get a schedule $\sigma_1$. Scale up each
communication delay $\hat{c}_{ij}$ in $\sigma_1$ by $R$ to get $\sigma_2$.
\end{algorithm}
\begin{theorem}
\label{thm:lgcomm-m}
$\sigma_2$ is a $\frac{2(R+1)}{3}$-approximate schedule on $m \ge n$ processors.
\end{theorem}
\begin{proof}
Let $\mu^*_1$ be the optimal makespan for $\hat{c}_{ij}$ and $\mu^*$ the optimal
makespan for $c_{ij}$. It is clear that $\mu^*_1 \le \mu^*$; otherwise, when we
scale down by $R$ we could get a better solution.

Let $\mu_1 = \mu(\sigma_1)$ and $\mu = \mu(\sigma_2)$. We have $\mu_1 \le
\frac{4}{3} \mu^*_1$ by algorithm \ref{alg:smcomm-mgen}. $\mu_1$ is the length
of some critical path $P$ of tasks in $\sigma_1$; $P$ has $g$ tasks and $g-1$
communication delays (some of which may be 0.)

Let $A$ be the contribution of task durations to $\mu_1$ and $B$ be the
contribution of delays. For each task $k$ in $P$ and each delay $\hat{c}_{ij}$
in $P$, we have $\hat{c}_{ij} \le d_k$. Therefore, $A \ge \mu_1/2$. When we
scale the $\hat{c}_{ij}$ up by $R$ to get $c_{ij}$ in $\sigma_2$, we get $\mu
\le A + R \, B \le \frac{1}{2} \mu_1 + \frac{R}{2} \mu_1 \le \frac{2(R+1)}{3}
\mu^*_1$, so $\mu \le \frac{2(R+1)}{3} \mu^*$.
\end{proof}
\section{Conclusions}
We have shown that for several scheduling problems, using convex programming, we
can obtain approximation bounds when energy constraints are present that are no
worse than the existing bounds obtained when energy is not considered. Our
analyses for the most part used the analyses of algorithms for the corresponding
energy-blind cases, adjusted where needed to fit the convex programming
formulation.

One possible direction for future work is to characterize the conditions
necessary for the approach in this paper to be applicable. We rely heavily on
the fact that, once the time/energy allocations are determined, our problem is
essentially an instance of the energy-blind problem, and can be solved using
energy-blind methods.

As an example, recomputation is a consideration for which our approach appears
to break. In scheduling models that permit multiple copies of tasks to be
computed on different processors, such as the model considered by Papadimitriou
and Yannakakis in \cite{PY1990}, the energy for each copy should be taken into
account. More work is necessary to determine the extent to which our methods are
applicable to these models. In particular, a convex programming formulation that
would permit energy constraints to be appended is not obvious and would be
interesting.
\bibliographystyle{plain}
\bibliography{ds4e}
\end{document}